\documentclass[12pt]{article}

\usepackage{graphicx}
\usepackage{pgfplots}
\usepackage{booktabs}
\usepackage{amsmath}
\usepackage{algorithm}
\usepackage{algorithmic}
\usepackage{amssymb}
\usepackage{amsfonts}
\usepackage{latexsym}
\usepackage{color}
\usepackage{subfigure}
\usepackage{latexsym}
\usepackage{graphicx,epstopdf}
\usepackage{subfigure}
\usepackage{authblk}
\usepackage{caption}
\usepackage{dsfont}

\newcommand{\be}{\begin{equation}}
\newcommand{\en}{\end{equation}}
\newcommand{\bea}{\begin{eqnarray}}
\newcommand{\ena}{\end{eqnarray}}
\newcommand{\beano}{\begin{eqnarray*}}
\newcommand{\enano}{\end{eqnarray*}}
\newcommand{\bee}{\begin{enumerate}}
\newcommand{\ene}{\end{enumerate}}

\newcommand{\Hil}{{\cal H}}

\newcommand{\Lc}{{\cal L}}

\newcommand{\Sc}{{\cal S}}

\newtheorem{thm}{Theorem}

\newtheorem{prop}[thm]{Proposition}

\newenvironment{proof}{\noindent {\bf Proof:}}{\hfill$\Box$}
\begin{document}

\title{A no-go result for the quantum damped harmonic oscillator}

\author{F. Bagarello$^{1,2}$, F. Gargano$^{1}$, F. Roccati$^3$\\

\small{
$^1$DEIM -  Universit\`a di Palermo,
Viale delle Scienze, I--90128  Palermo, Italy,}\\
\small{ $^2$I.N.F.N -  Sezione di Napoli,}\\
\small{
$^3$Dipartimento di Fisica e Chimica, Universit\`a degli Studi di Palermo, 
via Archirafi 36, I-90123 Palermo, Italy.}\\
\small{\emph{Email addresses:}\\
fabio.bagarello@unipa.it,
francesco.gargano@unipa.it,
federico.roccati@unipa.it}
}

\date{}
\maketitle
\begin{abstract}
In this letter we show that it is not possible to set up a canonical quantization for the damped harmonic oscillator using the Bateman lagrangian. In particular, we prove that no square integrable vacuum exists for the {\em natural} ladder operators of the system, and that the only vacua can be found as distributions. This implies that the procedure proposed by some authors is only formally correct, and requires a much deeper analysis to be made rigorous.
\end{abstract}


\section{Introduction}\label{sec:intro}

The problem of quantizing dissipative systems, and the damped harmonic oscillator (DHO) in particular, has a long story. {There are two main approaches}: (i) a time-independent lagrangian, proposed by Bateman and studied by many authors along the years, see \cite{bate,vitiello,dekk,dekk2,fesh, review2} and references therein, which is based on the idea that the damped oscillator is associated to a second, amplified, oscilllator (AHO), which gains all the energy lost by the damped one; (ii) an explicitly time-dependent lagrangian, see \cite{baldiotti,cheng,dekk,review2} and references therein, which reproduce the equation of motion of the DHO without the need of introducing its amplified counterpart. More recently, another non standard, time-independent,  lagrangian has been introduced for the DHO, where it was shown that there is no need to introduce any dual oscillator. However~\cite{osc1d}, in this case the quantization proved to be quite hard, if not entirely impossible, in the sense that the Schr\"odinger equation admits no exact solution, apparently.

In what follows, we will concentrate on the Bateman's approach, and we will show that quantizing the system is, in fact, impossible within the realm of Hilbert spaces. More in detail, we will be able to diagonalize the Hamiltonian for the system in terms of pseudo-bosonic operators, but we will also prove that, contrarily to what stated in some contributions in the literature, it is not possible to use these ladder operators to construct a biorthogonal set in some Hilbert space $\Hil$, simply because there is no vacuum in any $\Hil$ for the annihilation operators constructed along the way.  This impossibility, which was somehow recognized by many authors, including one of us, \cite{bag1,bag2}, and discussed as a {\em lack of normalization} of wave-functions, in fact leaves open the possibility of working in the space of distributions.

\section{The model and its problems}\label{sect2}

The classical equation for the DHO is $m\ddot x+\gamma \dot x+kx=0$, in which $m,\gamma$ and $k$ are the physical positive quantities of the oscillator: the mass, the friction coefficient and the spring constant. The Bateman lagrangian is
\be L=m\dot x\dot y+\frac{\gamma}{2}(x\dot y-\dot xy)-kxy,
\label{21}\en
which other than the previous equation, produces also $m\ddot y-\gamma \dot y+ky=0$, the differential equation associated to the AHO. In \cite{fesh}, this is called the {\em time reverse} equation of that for $x(t)$, for obvious reasons. The conjugate momenta are $$p_x=\frac{\partial L}{\partial \dot x}=m\dot y-\frac{\gamma}{2}\,y,\qquad
p_y=\frac{\partial L}{\partial \dot y}=m\dot x+\frac{\gamma}{2}\,y,
$$
and the corresponding classical Hamiltonian is
\be
H=p_x\dot x+p_y \dot y-L=\frac{1}{m} p_xp_y+\frac{\gamma}{2m}\left(yp_y-xp_x\right)+\left(k-\frac{\gamma^2}{4m}\right)xy.
\label{22}\en
By introducing the new variables $x_1$ and $x_2$ through 
\be
x=\frac{1}{\sqrt{2}}(x_1+x_2), \qquad y=\frac{1}{\sqrt{2}}(x_1-x_2),
\label{23}\en
$L$ and $H$ can be written as follows:
$$
L=\frac{m}{2}(\dot x_1^2-\dot x_2^2)+\frac{\gamma}{2}(x_2\dot x_1-x_1\dot x_2)-\frac{k}{2}(x_1^2-x_2^2)
$$
and
$$
H=\frac{1}{2m}\left(p_1-\frac{\gamma}{2}x_2\right)^2-\frac{1}{2m}\left(p_2-\frac{\gamma}{2}x_1\right)^2+\frac{k}{2}(x_1^2-x_2^2),
$$
where $p_1=\cfrac{\partial L}{\partial \dot x_1}=m\dot x_1+\cfrac{\gamma}{2}\,x_2$ and $p_2=\cfrac{\partial L}{\partial \dot x_2}=m\dot x_2-\cfrac{\gamma}{2}\,x_1$. By putting $\omega^2=\cfrac{k}{m}\,-\cfrac{\gamma^2}{4m^2}$ we can rewrite $H$ as follows:
\be
H=\left(\frac{1}{2m}p_1^2+\frac{1}{2}m\omega^2x_1^2\right)-\left(\frac{1}{2m}p_2^2+\frac{1}{2}m\omega^2x_2^2\right)-\frac{\gamma}{2m}(p_1x_2+p_2x_1).
\label{24}\en
In this section we will mostly consider $\omega^2>0$. The case {$\omega^2\leq0$}  will be briefly discussed later.

Following \cite{nakano} we impose the following canonical quantization rules between $x_j$ and $p_k$: $[x_j,p_k]=i\delta_{j,k}\mathds{1}$, working in unit $\hbar=1$. Here $\mathds{1}$ is the identity operator. This is equivalent to the choice in \cite{fesh}. Ladder operators can now be easily introduced:
\be a_k=\sqrt{\frac{m\omega}{2}}\,x_k+i\sqrt{\frac{1}{2m\omega}}\,p_k,
\label{25}\en
$k=1,2$. These are bosonic operators since they satisfy the canonical commutation rules: $[a_j,a^\dagger_k]=\delta_{j,k}\mathds{1}$. 

\vspace{2mm}

{\bf Remark:--} because of the unboundedness of  the operators involved, these commutators should be properly defined. For instance,both $[a_j,a^\dagger_k]$ and $[x_j,p_k]$ are well defined on Schwartz test functions: $[a_j,a^\dagger_k]\varphi(x)=\varphi(x)$, for all $\varphi(x)\in\Sc(\mathbb{R})$.

\vspace{2mm}

In terms of these operators the quantum version of the Hamiltonian $H$ in (\ref{24}) can be written as
\be
\left\{
\begin{array}{ll}
	H=H_0+H_I,\\
H_0=\omega\left(a_1^\dagger a_1-a_2^\dagger a_2\right),\\
	H_I=\cfrac{i\gamma}{2m}\left(a_1a_2-a_1^\dagger a_2^\dagger\right)\\
\end{array}
\right.
\label{26}\en 
In \cite{dekk,fesh} this Hamiltonian is diagonalized by using the $QU(2)$ algebra. The surprise is that, even if $H$ is (at least formally) self-adjoint, $H=H^\dagger$, its eigenvalues appear to be complex. This is quite strange, and in \cite{fesh} the authors claim that this is due to the fact that $H_I$ is only formally Hermitian {\em because the normalization integral for the eigenstates...is infinite, a result that follows from the fact that the eigenvalues are imaginary}. This is, for us, not  really a good explanation: we get imaginary eigenvalues because $H_I$ is not Hermitian, and this happens because the eigenstates cannot be normalized, and this is due to the existence of imaginary eigenvalues! This argument sounds a little bit tautological. For this reason, we propose here a different approach to the analysis of $H$, based on the generalized Bogoliubov transformation considered in \cite{nakano}. However, as we will show, while this transformation is useful {to diagonalize} $H$ making no reference to $QU(2)$ (which will be replaced in the following by formal pseudo-bosonic ladder operators, \cite{baginbagbook}), the conclusions deduced in \cite{nakano} are wrong. We will return on this point later.

Following \cite{nakano} we introduce the following operators:
\be
A_1=\frac{1}{\sqrt{2}}(a_1-a_2^\dagger), \quad A_2=\frac{1}{\sqrt{2}}(-a_1^\dagger+a_2),
\label{27}\en
as well as 
\be
B_1=\frac{1}{\sqrt{2}}(a_1^\dagger+a_2), \quad B_2=\frac{1}{\sqrt{2}}(a_1+a_2^\dagger).
\label{28}\en
These operators satisfy the following requirements:
\be
[A_j,B_k]=\delta_{j,k}\mathds{1},
\label{29}
\en
with $B_j\neq A_j^\dagger$, $j=1,2$. Moreover, $A_1=-A_2^\dagger$ and $B_1=B_2^\dagger$. The fact that $B_j\neq A_j^\dagger$ follows from the fact that the one in (\ref{27})-(\ref{28}) is not a Bogoliubov transformation, but only a generalized version of it, \cite{bagfri2}. Incidentally, in \cite{nakano} it is also given a second transformation, $A_1=\frac{1}{\sqrt{2}} \, (a_1-a_2^\dagger)=A_2^\dagger$, $B_1=\frac{1}{\sqrt{2}} \, (a_1^\dagger-a_2)=-B_2^\dagger$, which we could also use here. However, since our analysis does not change, we will restrict to the transformation in  (\ref{27})-(\ref{28}). Notice that this map is clearly reversible, since $a_j$ and $a_j^\dagger$ can be written in terms of $A_j$ and $B_j$. 

In \cite{baginbagbook} operators of this kind were analyzed in detail, producing several mathematical results (mainly on unbounded operators and biorthogonal sets of vectors), and were shown to appear often in quantum models of PT-quantum mechanics,  \cite{ben,mosta}. The main idea is that, for operators like these, we can extend the usual ladder construction used for bosons, paying some price (like, quite often, the validity of the basis property). We refer to \cite{baginbagbook} for this and other aspects of pseudo-bosonic operators.

In terms of these operators $H$ can now be written as follows:
\be
\left\{
\begin{array}{ll}
	H=H_0+H_I,\\
	H_0=\omega\left(B_1A_1-B_2A_2\right),\\
	H_I=\cfrac{i\gamma}{2m}\left(B_1A_1+B_2A_2+\mathds{1}\right),\\
\end{array}
\right.
\label{210}\en 
which only depends on the pseudo-bosonic number operators $N_j=B_jA_j$, \cite{baginbagbook}. This is exactly the same Hamiltonian found in \cite{nakano}, and it is equivalent to that given in \cite{dekk,fesh} and in many other papers on this subject. In \cite{nakano}, the authors introduce the vacuum for the annihilation operators $A_1$ and $A_2$ as the action of an unbounded operator on the vacuum of $a_1$ and $a_1$, and they construct new vectors out of this vacuum, claiming that these vectors, all together, form a Fock basis with norm equal to one. The next theorem, which is the main result of this letter, shows that this is entirely wrong and suggest a possible way out to solve this difficulty. 

\begin{prop}\label{prop1}
	There is no non-zero function $\varphi_{00}(x_1,x_2)$ satisfying $$A_1\varphi_{00}(x_1,x_2)=A_2\varphi_{00}(x_1,x_2)=0.$$ Also, there is no non-zero function $\psi_{00}(x_1,x_2)$ satisfying $$B_1^\dagger\psi_{00}(x_1,x_2)=B_2^\dagger\psi_{00}(x_1,x_2)=0.$$
\end{prop} 

\begin{proof}
Let us assume that a non-zero function $\varphi_{00}(x_1,x_2)$ satisfying $A_1\varphi_{00}(x_1,x_2)=A_2\varphi_{00}(x_1,x_2)=0$ does exist. Hence we should have also $(A_1-A_2)\varphi_{00}(x_1,x_2)=0$ and $(A_1+A_2)\varphi_{00}(x_1,x_2)=0$, so that
$$
(x_1-x_2)\varphi_{00}(x_1,x_2)=0, \qquad (\partial_1+\partial_2)\varphi_{00}(x_1,x_2)=0.
$$
It is clear that there is no non-zero solution of the first equation. The only solution is a distribution: $\varphi_{00}(x_1,x_2)=\alpha\delta(x_1-x_2)$, $\alpha\in\mathbb{C}$.

The proof for $\psi_{00}(x_1,x_2)$ is completely similar: computing $(B_1^\dagger+B_2^\dagger)\psi_{00}(x_1,x_2)=0$ and
 $(B_1^\dagger-B_2^\dagger)\psi_{00}(x_1,x_2)=0$ we easily get
 $$
 (x_1+x_2)\psi_{00}(x_1,x_2)=0, \qquad (\partial_1-\partial_2)\psi_{00}(x_1,x_2)=0,
 $$
 for which exists only the weak solution $\psi_{00}(x_1,x_2)=\beta\delta(x_1+x_2)$, $\beta\in\mathbb{C}$.

\end{proof}

\vspace{2mm}

We see that the proof of the proposition is extremely easy and, because of this, particularly interesting. Notice also that, compared with other papers on the same subject, see \cite{bag1,bag2} for instance, we are not assuming here that $\varphi_{00}(x_1,x_2)$ and $\psi_{00}(x_1,x_2)$ belong to some specific Hilbert space. For this reason, our result does not depend on the metric we can introduce in $\Lc^2(\mathbb{R}^2)$ to take care of possible divergences in the norms of the eigenfunctions of $H$: stated differently, replacing $\left<f,g\right>=\int_{\mathbb{R}^2}\overline{f(x_1,x_2)}\,g(x_1,x_2)\,dx_1\,dx_2$ with some $\left<f,g\right>_w=\int_{\mathbb{R}^2}\overline{f(x_1,x_2)}\,g(x_1,x_2)\,w(x_1,x_2)\,dx_1\,dx_2$, for any choice of weight function $w(x_1,x_2)$, does not affect our proposition.

\subsection{The {overdamped} case, $\omega^2<0$}
The previous results have been obtained under the constraint $\omega^2>0$ which allowed to define the bosonic operators \eqref{25}.
In the  $\omega^2<0$ case, that for simplicity we assume corresponding to $\omega=e^{i\pi/2}\,\tau,\,\tau\in\mathbb{R}_+$, easy computations show that the commutators satisfy $[a_j,a_k^\dag]=i\delta_{j,k}\mathds{1}$, and hence they are not bosonic operators.

This issue is solved introducing the following operators\footnote{The principal square root of $\omega$ has been chosen.} 
\bea a_k=e^{i\pi/4}\sqrt{\frac{m\tau}{2}}x_k+ie^{-i\pi/4}\sqrt{\frac{1}{2m\tau}}p_k,\quad b_k= e^{i\pi/4}\sqrt{\frac{m\tau}{2}}x_k-ie^{-i\pi/4}\sqrt{\frac{1}{2m\tau}}p_k,
\label{211}\ena
$k=1,2$..
Of course, $a_k^\dag\neq b_k, \, k=1,2$, but these operators satisfy the \textit{extended} canonical commutation rules $[a_j,b_k]=\delta_{j,k}\mathds{1}$, which means that  $a_k,b_k,\, k=1,2$, define two pairs of formal pseudo-bosonic operators, \cite{baginbagbook}.

We can now use $b_k$ in place of $a_k^\dagger$ in \eqref{27}-\eqref{28} and define the operators
\bea
A_1&=&\frac{1}{\sqrt{2}}(a_1-b_2), \quad A_2=\frac{1}{\sqrt{2}}(-b_1+a_2),\\
B_1&=&\frac{1}{\sqrt{2}}(b_1+a_2), \quad B_2=\frac{1}{\sqrt{2}}(a_1+b_2),
\label{212}\ena
which again satisfy $[A_j,B_k]=\delta_{j,k}\mathds{1}$.
With this construction the Hamiltonian has the same form and properties of \eqref{210}, and the results given in Proposition \ref{prop1} are still valid.

{\bf Remarks:--} (1) the case $\omega^2=0$ is not really interesting for us, in the present context, since in this case the use of bosonic or pseudo-bosonic operators is quite unlikely, {being} $H$   (\ref{24})  no longer quadratic in $x_j$. 

(2) It is possible to check that nothing really changes if we consider a quantized non commutative version of the Hamiltonian in  (\ref{24}), i.e. if we assume that $[x_1,x_2]=i\theta$ and $[p_1,p_2]=i\nu$, other than having $[x_j,p_j]\neq0$. Using the same ideas as in \cite{bagfri} and  adopting a Bopp shift, we recover a similar no-go result.

\section{Conclusions}

We have seen that the idea of working in an Hilbert space of square integrable functions when dealing with the Bateman lagrangian for the DHO does not work.

Also, an explicitly time-dependent Hamiltonian of the kind proposed by Caldirola and by Kanai independently, \cite{Caldi, dekk2,Kanai}, restores square-integrability of the eigenstates, which can be exactly computed, \cite{baldiotti}, but again produce {\em strange} imaginary eigenvalues for the formally self-adjoint Caldirola-Kanai Hamiltonian $H_{CK}$, which are explained considering the domain of the operator $xp_x+p_xx$. Moreover, as discussed for instance in \cite{crespo}, the role of $H_{CK}$ is not fully accepted in the context of the DHO: in fact, in \cite{crespo} one of the main conclusions is that $H_{CK}$ {\em does not describe ...a damped harmonic oscillator of mass $m_0$ ...but a particle of mass $m(t)=m_0e^{\alpha t}$ subject to a force $F(t)= m_0e^{\alpha t}x$....}. This is just a small evidence of the fact that quantization of dissipative systems is an hard topic, which is still not completely understood.


{
 The main result of this letter suggests that the possible way-out to overcome the mathematical issues arising in the quantization of the DHO is to reconsider the system in a distributional setting. Of course, one needs to define the proper ladder operators used to write the Hamiltoanian, control the notion of biorthogonality of the eigenstates, and check whether this approach has some concrete usefulness for the analysis of the DHO. This is work in progress.
}

\section*{Acknowledgements}
The authors acknowledge partial support from Palermo University. F.B. and F.G.  acknowledge partial support from G.N.F.M. of the I.N.d.A.M.
F.G. acknowledges support by M.I.U.R.

\end{document}